\documentclass[12pt]{article}
\usepackage{amsfonts}
\usepackage{amsmath}
\usepackage{amsthm}
\usepackage{mathrsfs}
\usepackage{fullpage}
\usepackage{url}

\newtheorem{theorem}{Theorem}

\newtheorem{lemma}{Lemma}

\newcommand{\Del}{\ensuremath{\nabla}}
\newcommand{\M}{\ensuremath{\mathbb M}}
\newcommand{\Lie}{\pounds}
\newcommand{\Scri}{\ensuremath{\mathscr I} }
\newcommand{\pback}[1]{\ensuremath{ \underleftarrow{#1}}}
\newcommand{\XX}{{\Gamma}}

\begin{document}


\title{
{\bfseries Covariant Derivatives on Null Submanifolds}
}

\author{
	Don Hickethier \\[-2.5pt]
	\normalsize
	\textit{Department of Mathematics, Flathead Valley Community College,
		Kalispell, MT  59901} \\[-2.5pt]
	\normalsize
	{\tt dhicketh{\rm @}fvcc.edu}
		\and
	Tevian Dray \\[-2.5pt]
	\normalsize
	\textit{Department of Mathematics, Oregon State University,
		Corvallis, OR  97331} \\[-2.5pt]
	\normalsize
	{\tt tevian{\rm @}math.oregonstate.edu} \\
}

\date{\normalsize July 26, 2011}

\maketitle

\begin{abstract}
The degenerate nature of the metric on null hypersurfaces makes it difficult
to define a covariant derivative on null submanifolds.  Recent approaches
using decomposition to define a covariant derivative on null hypersurfaces are
investigated, with examples demonstrating the limitations of the methods.
Motivated by Geroch's work on asymptotically flat spacetimes, conformal
transformations are used to construct a covariant derivative on null
hypersurfaces, and a condition on the Ricci tensor is given to determine when
this construction can be used.  Several examples are given, including the
construction of a covariant derivative operator for the class of spherically
symmetric hypersurfaces.
\end{abstract}

\section{Introduction}
\label{intro}

Given a null submanifold of a Lorentzian spacetime, is it possible to define a
preferred torsion-free, metric-compatible covariant derivative?  Is it
possible to determine, a priori, when such a connection can be found?

One place where the need for a such a derivative arises is in the study of
asymptotically flat spacetimes.  According to Geroch~\cite{rG76}:
\begin{quote}
\textit{In the null case, one has no unique derivative operator, and so one
works more with Lie and exterior derivatives, and with other differential
concomitants.  As a general rule, it is considerably more difficult in the
null case to write down formulae which say what one wants to say.  }
\end{quote}
Geroch's treatment of the null boundaries of asymptotically flat
spacetimes~\cite{rG76} motivates the techniques used here to find a preferred
derivative operator on (some) null submanifolds.

We briefly review traditional approaches to this problem in
Section~\ref{Traditional}, and summarize the relevant parts of Geroch's
construction of a null asymptotic boundary in Section~\ref{geroch}.
Section~\ref{first-theorem} establishes an existence condition, showing when
the Geroch construction can be used to construct a preferred covariant
derivative on a given null submanifold, and Section~\ref{theorem-two} then
gives a simpler condition on the Ricci tensor for determining when this
construction is possible.  We present several examples in
Section~\ref{Examples}, including the horizon of the Schwarzschild geometry,
and discuss our results in Section~\ref{Summary}.

\section{Traditional Approaches}
\label{Traditional}

\subsection{Gauss decomposition}

Let $(M,g)$ be a \textit{spacetime}, that is a manofild $M$ together with a
nondegenerate metric $g$ of Lorentzian signature.  If $(\Sigma,q)$ is a
submanifold of $(M,g)$ given by $\varphi:\Sigma\rightarrow M$, and if
$q=\varphi^* g$ is a nondegenerate metric on $\Sigma$, then a connection
$\Del$ on $M$ induces a natural connection $D$ on $\Sigma$.  A traditional
approach to defining this connection is to split $TM$ into the direct sum
 \begin{equation}
   TM=T\Sigma\oplus T\Sigma^{\perp},
 \end{equation}
where $T\Sigma^{\perp}$ is the orthogonal complement of $T\Sigma$ in
$TM$.

For $X,\,Y\in \XX (TM)$, $\Del_{X} Y$ can be separated on $\Sigma$
into tangential and orthogonal components of $TM$ which define the
induced connection, $D_X\,Y$, and the second fundamental form,
$II(X,\,Y)$.  Explicitly, we have
\begin{align}
  D_X\,Y & = (\Del_X\,Y)^{\parallel}\\
  II(X,\,Y) & = (\Del_X\,Y)^{\perp} = \Del_{X} Y-D_X Y.
\end{align}
If $\Del$ is the Levi-Civita connection on $M$, then $D$ turns out to be the
Levi-Civita connection on $\Sigma$.  The decomposition
 \begin{equation}
\label{gauss-formula}
   \Del_{X}\,Y = D_X Y + II(X,Y)
 \end{equation}
is called Gauss' formula~\cite{mS79}.

\subsection{Duggal decomposition}
\label{Duggal}

Difficulties arise when the metric $q$, on $\Sigma$, is degenerate.
Furthermore, if $\Sigma$ is lightlike, $TM$ cannot be decomposed into the
direct sum of $T\Sigma$ and $T\Sigma^{\perp}$, since there are vectors in
$T\Sigma$ that are also in $T\Sigma^{\perp}$, as well as vectors that are in
neither space.  Despite these difficulties, Duggal and Bejancu~\cite{DB96}
(henceforth referred to as \textit{Duggal}) introduced a decomposition that
produces equations similar to the Gauss formula~\eqref{gauss-formula}, as we
now describe.

Given a  lightlike submanifold $\Sigma$ of $M$ with tangent space
$T\Sigma$, the goal is to create a decomposition of $TM$ by
producing a vector bundle similar to $T\Sigma^{\perp}$.  Choose a
\emph{screen manifold} $Scr(T\Sigma)\subset T\Sigma$ such that
 \begin{equation}
   T\Sigma = Scr(T\Sigma) \oplus T\Sigma^{\perp}.
 \end{equation}

Given a screen manifold, $Scr(T\Sigma)$, Duggal proves the existence of
a unique complementary vector bundle, $tr(T\Sigma)$, to $T\Sigma$, called
the \emph{lightlike transversal vector bundle} of $\Sigma$ with respect
to $Scr(T\Sigma)$.
\begin{theorem}[Duggal]
\label{DuggalThm}
Let $(\Sigma,\,q,\,Scr(T\Sigma))$ be a lightlike hypersurface of a Lorentzian
manifold $(M,g)$.  Then there exists a unique vector bundle
$tr(T\Sigma)\subset TM$, of rank 1 over $\Sigma$, such that for any nonzero
$\xi\in \XX (T\Sigma^\perp)$ there exists a unique $N \in \XX (tr(T\Sigma))$
such that
 \begin{equation}
 \label{N_properties}
   N\cdot\xi =-1,\quad N\cdot N=0,\quad N\cdot W
   =0\,\forall\, W\in Scr(T\Sigma).
 \end{equation}
\end{theorem}

Thus, $tr(T\Sigma)\perp Scr(T\Sigma)$ and since $tr(T\Sigma)$ is
$1-$dimensional, $\XX (tr(T\Sigma))=Span( N )$.  By construction we
have $tr(T\Sigma)\cap T\Sigma = \{0\}$ and have decomposed $TM$
to
\begin{equation}
\label{tr_decomp}
  TM = Scr(T\Sigma)\oplus(T\Sigma^{\perp}\oplus tr(T\Sigma) )=
  T\Sigma\oplus tr(T\Sigma)
\end{equation}
where $TM$ is restricted to $\Sigma$.

We can use~\eqref{tr_decomp} to decompose the connection $\Del$
on $M$ as follows.  Let \hbox{$X,\,Y\in\XX (T\Sigma)$} and $V\in\XX
(tr(\Sigma))$.  Then, since $tr(T\Sigma)$ has $rank\, 1$, we can write
\begin{align}
  \Del_X\,Y & = D_X\,Y + B(X,Y)\,N \label{gauss} \\
  \Del_X\,V & = -A_N\,X + \tau(X)\,N, \label{weingarten}
\end{align}
where $D_X\,\,Y,A_N\,X \in \XX(T\Sigma)$.  Equation
\eqref{gauss} can be thought of as the Gauss formula for the
lightlike hypersurface and~\eqref{weingarten} as the lightlike
Weingarten formula.  Under this decomposition, $D_X\,Y$ is a
connection on $\Sigma$, but, as discussed in Duggal~\cite{DB96},
this connection is not, in general, metric-compatible.

\subsection{Example}

To investigate the Duggal decomposition, consider the line element
\begin{align}
  ds^2 = -2du\,dv + q_{ij}\,dx^i\,dx^j
\label{example}
\end{align}
Choose a screen $Scr(T\Sigma)=Span\left(\{X_1,\,X_2\}\right)$ by
setting
\begin{align}
\label{X_k}
  X_k = \frac{\partial}{\partial\,x^k} + \alpha_k\,\xi
\end{align}
where $\xi=\eta\,\frac{\partial}{\partial\,v}\in\XX (T\Sigma^{\perp})$ and $\alpha_k$ is a function of $v,\,x^1\text{ and } x^2$.  All possible screens can be obtained by choosing different
$\alpha_k$.  For $N\in \XX (tr(T\Sigma))$ satisfying Theorem
\ref{DuggalThm}, the vectors $\{\xi,X_1,X_2,N\}$ form a basis
for $\XX (TM)$.  As shown in Duggal~\cite{DB96}, the covariant derivatives take
the form
\begin{subequations}
\begin{align}
  \Del_{X_j}\,X_i & = \gamma^0\,_{ij}\,\xi + \gamma^k\,_{ij}\,X_k + B_{ij}\,N  \label{gauss1}\\
  \Del_{X_j}\, \xi & = \gamma^0\,_{0j}\,\xi + \gamma^k\,_{0j}\,X_k \label{gauss2}\\
  \Del_{\xi}\, X_i & = \gamma^0\,_{i0}\,\xi + \gamma^k\,_{i0}\,X_k \label{gauss3}\\
  \Del_{\xi}\, \xi & = \gamma^0\,_{00}\,\xi \label{gauss4}
\end{align}
\end{subequations}
and
\begin{subequations}
\begin{align}
  \Del_{X_j}\, N & = -A^k\,_j\,X_k + \tau_j\,N \label{wein1}\\
  \Del_{\xi}\, N & = -A^k\,_0\,X_k + \tau_0\,N \label{wein2}
\end{align}
\end{subequations}
Duggal defines the last term in equation~\eqref{gauss1} to be the
second fundamental form,
\begin{equation}
II(X,Y)=B(X,Y)\,N
\end{equation}
Thus, equations \eqref{gauss1}--\eqref{gauss4} decompose $\Del$ to a form
similar to Gauss' formula~\eqref{gauss}.  Once the coefficients
in~\eqref{gauss1}--\eqref{gauss4} are known, a connection $D$ on $\Sigma$ has
been constructed.  The properties of the Levi-Civita connection $\Del$ on
$\Sigma$ can be used to show that
\begin{subequations}
\begin{align}
\label{duggal_connect}
  \gamma^k\,_{ij} & = \frac12\,q^{kh}\bigl(
	X_j(q_{ih}) + X_i(q_{hj}) - X_h(q_{ij}) \bigr) \\
  \gamma^0\,_{ij} & = -g_{ik}\,A^k\,_j \\
  \gamma^0\,_{0j} & = -\tau_j \\
  \gamma^k\,_{0j} & = \frac12\,g^{ik}\,\xi(g_{ij}) \\
  \gamma^0\,_{00} & = -\tau_0 \\
  B_{ij} & = \frac12\, \xi(g_{ij}).
\end{align}
\end{subequations}
and the induced covariant derivative on $T\Sigma$ becomes
\begin{subequations}
\begin{align}
  D_{X_j}\,X_i & = \gamma^0\,_{ij}\,\xi + \gamma^k\,_{ij}\,X_k\\
  D_{X_j}\, \xi & = \gamma^0\,_{0j}\,\xi + \gamma^k\,_{0j}\,X_k\\
  D_{\xi}\, X_i & = \gamma^0\,_{i0}\,\xi + \gamma^k\,_{i0}\,X_k\\
  D_{\xi}\, \xi & = \gamma^0\,_{00}\,\xi
\end{align}
\end{subequations}

\subsection{Uniqueness}
\label{Unique}

In general, the above construction of the induced covariant derivative $D$
depends on the choice of screen, that is, on the choice of $X_k$.  However,
Duggal further proves that under certain conditions there is a unique induced
connection on $\Sigma$.

\begin{theorem}[Duggal]
\label{duggal_IIvanish}
Let $(\Sigma,\,q,\,Scr(T\Sigma))$ be a lightlike hypersurface of $(M,g)$.
Then the induced connection $D$ is unique, that is, $D$ is independent of
$Scr(T\Sigma)$, if and only if the second fundamental form $II$ vanishes
identically on $\Sigma$.  Furthermore, in this case, $D$ is torsion free and
metric compatible.
\end{theorem}

Theorem~\ref{duggal_IIvanish} implies that if $B_{ij} \neq 0$, or
equivalently, $\xi(g_{ij})\neq 0$ for all $i,\,j$, then there is a need for a
new method to define a covariant derivative on $\Sigma$.

Returning to our example~\eqref{example}, if the 2-metric $q_{ij}$ is that of
a plane, that is, if we consider a null plane $\Sigma=\{u=0\}$ in Minkowski
space $\M^4$ in null rectangular coordinates with line element
 \begin{align}
   ds^2 = -2\,du\,dv+dx^2+dy^2,
 \end{align}
then it is straightforward to show that $B_{ij}=0$, so that we obtain a unique
connection on $\Sigma$ regardless of the screen chosen.  It is a useful
exercise to check this explicitly, using
\begin{align}
    X_1=\alpha\,\frac{\partial}{\partial v} + \frac{\partial}{\partial x},\quad X_2=\beta\,\frac{\partial}{\partial v} +
    \frac{\partial}{\partial y},\quad \xi =\eta\,\frac{\partial}{\partial
    v}
 \end{align}
with arbitrary $\alpha$, $\beta$, $\eta$, which implies
 \begin{equation}
   N=\frac{1}{\eta}\left(\frac{\partial}{\partial u} + \left(\frac{\alpha^2+\beta^2}{2}\right)\,\frac{\partial}{\partial v} +  \alpha
\frac{\partial}{\partial x} + \beta\frac{\partial}{\partial
y}\right)
 \end{equation}

However, if the 2-metric $q_{ij}$ is that of a sphere, that is, if we consider
a null cone $\Sigma=\{u=0\}$ in Minkowski space $\M^4$ in null spherical
coordinates with line element
\begin{equation}
ds^2 = -2\,du\,dv+r^2\,d\theta^2+r^2\,\sin^2\theta\, d\phi^2
\label{cone}
\end{equation}
with $r=(v-u)/\sqrt{2}$, then $B_{ij}\ne0$, since $q_{ij}$ depends on $r$.
Thus, Duggal's construction using the screen distribution and transversal
vector bundle does not yield a preferred, metric-compatible, torsion-free
connection on the null cone.

An alternate construction will be developed in this paper.

\subsection{Connections via the Pullback}
\label{Pullback}

Another possible way to construct a connection on a submanifold is to use
pullbacks.  Although this method often fails, we will show in subsequent
sections that a modified version of this method has wide applicability.

Let $\Sigma$ be a submanifold of $\M^4$ with $\varphi:\Sigma \rightarrow \M^4$
an embedding of $\Sigma$ into $\M^4$.  If $w_a$ is a 1-form on $\M^4$, its
\emph{pullback} $\varphi^*\, w_a$ is a 1-form on $\Sigma$.  We can therefore
attempt to define a covariant derivative $D$ on $\Sigma$ by pulling back the
covariant derivative operator $\Del$ on $\M^4$, that is, we seek an operator
$D$ satisfying
\begin{equation}
\label{pullback-formal}
D_a (\varphi^*w_b)=\varphi^*(\Del_a w_b).
\end{equation}

We adopt a less formal notation, and write $\pback{w_b}$ instead of
$\varphi^*w_b$ for the pullback of $w_b$ to~$\Sigma$.  With this new notation,
the pullback of the covariant derivative is written
\begin{equation}
D_a \pback{w_b}=\pback{\Del_a w_b}.
\label{pullback}
\end{equation}

It is easily checked that the Duggal connection on a null plane in Minkowski
space, as defined in Section~\ref{Duggal} and shown to be unique in
Section~\ref{Unique}, also satisfies~\eqref{pullback}, and is in fact uniquely
defined by this condition.  In general, however,~\eqref{pullback} alone is not
enough to determine a well-defined covariant derivative $D$ on $\Sigma$.


The problem is that there are many 1-forms $w$ with the same pullback
$\pback{w}$; for $D$ to be well-defined on $\Sigma$, it must not depend on
this choice.  If $\Sigma=\{u=0\}$, then
\begin{equation}
\pback{w+f\,du}=\pback{w}
\end{equation}
since $\pback{du}=0$.  Thus,~\eqref{pullback} will be well-defined
if (and only if)
\begin{equation}
\pback{\Del_X\, du}=0
\end{equation}
for all $X\in \XX(T\Sigma)$.

As an example, consider the null cone in Minkowski space given by $u=t-r=0$.
Then
\begin{equation}
\pback{\Del_X du} = \pback{\Del_X dt} - \pback{\Del_X dr} \ne 0
\end{equation}
since $\pback{\Del_X dt}=0$ but $\pback{\Del_X dr}\ne0$.  Thus, the pullback
method fails on the null cone.

In terms of coordinates $\{x^i\}$ on the surface $\Sigma=\{u=0\}$, extended to
a neighborhood of~$\Sigma$, it is easily seen that condition~\eqref{pullback}
for the existence of a well-defined pullback connection is equivalent to the
vanishing of the appropriate Christoffel symbols, namely
\begin{equation}
\Gamma^u{}_{ij} = 0
\label{Christoffel}
\end{equation}
In the case of the null cone, we have
\begin{equation}
\Gamma^u\,_{\theta\theta} = -\frac{v-u}{2} = -r
= \frac{1}{\sin^2\theta} \,\Gamma^u\,_{\phi\phi}
\end{equation}
The simple dependence of these terms on $r$ suggests a possible strategy:
remove the $r$-dependence by rescaling the line element by $r^2$, after which
these Christoffel symbols will vanish, and a well-defined covariant derivative
can be defined.

We implement this strategy in the remainder of the paper

\section{Asymptotically Flat Spacetimes}
\label{geroch}

There is a well-known context in general relativity for studying a null
submanifold, namely the construction of null infinity for an asymptotically
flat spacetime.  Since much of that construction will be useful in our more
general context, we briefly review it here.  Our presentation follows the
classic 1976 paper of Geroch~\cite{rG76}.


An \emph{asymptote} of a spacetime $(\widetilde M,\ \widetilde g_{ab})$ is a
manifold $M$ with boundary $\Scri$, together with a smooth Lorentzian metric
$g_{ab}$ on $M$, a smooth function $\Omega$ on $M$, and a diffeomorphism by
means of which we identify $\widetilde M$ and $M-\Scri$, satisfying the
following conditions:
\begin{enumerate}
\item
On $\widetilde M$, $g_{ab} = \Omega^{2}\, \widetilde g_{ab}$;
\label{AF:1}
\item
On $\Scri$,
$\Omega = 0$, \label{AF:2a}
$\Del_a\Omega \neq 0$, and \label{AF:2b}
$g^{ab}(\Del_a\Omega)( \Del_b \Omega)=0$,
\label{AF:2c}
\end{enumerate}
where $\Del_a$ denotes covariant differentiation on $M$.

\goodbreak

The Ricci curvature tensors of the conformally related metrics $g_{ab}$ and
$\widetilde g_{ab}$ are related by
\begin{align}
    \widetilde R_{ab}  =  R_{ab} & + (s-2)\, \Omega^{-1}\, \Del_a \Del_b\, \Omega
                + \Omega^{-1}\, g_{ab} \Del^m \Del_m\, \Omega \nonumber \\
                & - (s-1)\, \Omega^{-2}\, g_{ab}\, \left(\Del^{m} \Omega\right)
             \left(\Del_{m}\, \Omega\right)
\label{Ric}
\end{align}
where $s$ is the dimension of $M$.

We introduce the normal vector field
\begin{equation}
n^a = g^{ak}\, \Del_k \Omega = \Del^a \Omega
\end{equation}
and compute
\begin{equation}
\Lie_n\, g_{ab} = 2\, \Del_a \Del_b \Omega
\label{LiegOmega}
\end{equation}
Following Geroch~\cite{rG76}, we assume that the physical stress-energy tensor
vanishes asymptotically to order 2, that is, we assume that
$\Omega^{-2}\widetilde{R}^a{}_b$ admits a smooth extension to $\Scri$, which
in turn implies that $\Omega\widetilde{R}_{ab}$ is zero on $\Scri$.
\footnote{We will weaken this assumption below, which will affect the
numerical factor in~\eqref{ODE}.}
As shown by Geroch~\cite{rG76}, we can then use the gauge freedom in the
choice of $\Omega$ to ensure that the pullback to $\Scri$ of the RHS
of~\eqref{LiegOmega} vanishes.  Explicitly, by solving the ordinary
differential equation
\begin{equation}
n^c \Del_c \ln \omega = -\frac{1}{s} \>\Box\Omega
\label{ODE}
\end{equation}
along each integral curve of $n^a$ on $\Scri$, where
$\Box \Omega = g^{ab}\Del_a \Del_b\, \Omega$
is the d'Alembertian, and setting $\overline\Omega=\omega\Omega$, so that
\begin{eqnarray}
\overline g_{ab} &=& \omega^2\, g_{ab} \\
\overline n^a &=& \omega^{-1}\, n^a
\label{conformal}
\end{eqnarray}
then
\begin{equation}
\pback{\Lie_{\overline n}\, \overline g_{ab}} = 0
\label{KV}
\end{equation}

Thus, given an asymptotically flat spacetime satisfying the original Geroch
conditions, one can assume without loss of generality that $n$ is in fact a
Killing vector field.

A \textit{divergence-free conformal frame} $(g_{ab},n^c)$
satisfying~\eqref{KV} (where we have dropped the bars) has an additional
property: the pullback connection is well-defined on the null submanifold
$\Scri$.  To see this, we set
\begin{equation}
q_{ab} = \pback{g_{ab}}
\label{gback}
\end{equation}
and note that
\begin{equation}
\label{Lie-q}
\Lie_{n}\,q_{ab} = \pback{\Lie_{n}\, g_{ab}} = \pback{2\,\Del_a\, n_b}.
\end{equation}
Setting $u=\Omega$, the 1-form $n_b$ is just $du$, so that
\begin{equation}
\left(\Lie_n\,g\right)_{ij} =  -2\,\Gamma^u\,_{i j}.
\end{equation}
and the result now follows by comparison with~\eqref{Christoffel}.

\section{Covariant Derivatives on Null Submanifolds}
\label{first-theorem}

We now adapt the results from the previous section for asymptotically flat
spacetimes to more general null submanifolds.

Recall that a vector field $v^a$ is \emph{Killing} if $\Lie_v g_{ab}=0$.
Given a null surface, our first result is that if the normal vector is
Killing, then there is a well-defined covariant derivative on $\Sigma$.
Following Geroch, we then consider the conditions under which a conformally
related metrics admits a well-defined covariant derivative.  When appropriate
conditions are satisfied, we further propose that the resulting notion of
covariant derivative be regarded as the natural choice on the null
submanifold.

\begin{lemma}[Covariant Derivative on $\Sigma$]
\label{CovDerivI}
Let $\Sigma=\{u=0\}$ be a null submanifold of a given spacetime
($M$,$g_{ab}$), let $n_a=\Del_a u$, and let $q_{ab}$ be the induced degenerate
metric on $\Sigma$, as in~\eqref{gback}.  If $\Lie_n q_{ab}=0$ on $\Sigma$,
then the connection defined by the pullback, as in~\eqref{pullback}, is
well-defined.
\end{lemma}

\begin{proof}
Let $w_b$ be any 1-form on the surface $\Sigma$.  We would like to define a
covariant derivative using the pullback, $D_a\,w_b = \pback{\Del_a\,W_b}$,
where $W_b$ is a 1-form on $M$ such that $\pback{W_b}=w_b$, but we must show
that this is well defined.

On $\Sigma$, $\pback{n_b}=0$, since $n_b\,dx^b = du$.  Let $V_b=W_b+k\,n_b$
where $k$ is any function.  $V_b$ is the most general 1-form with the same
pullback as $W_b$,
\begin{equation}
\pback{V_b}=\pback{W_b}+\pback{k\,n_b}=w_b + 0.
\end{equation}
Consider the pullback of the derivative of $V_b-W_b$,
\begin{eqnarray}
\pback{\Del_a\,(V_b-W_b)}
  &=& \pback{\Del_a\,(k\,n_b)} \nonumber \\
  &=& \left(\pback{\Del_a\,k}\right)\,\pback{n_b}
	+ k\big|_{u=0}\,\left(\pback{\Del_a\,n_b}\right)\nonumber \\
  &=& \frac12 \,k\big|_{u=0}\,\left(\Lie_n q_{ab}\right)
   =  0
\end{eqnarray}
by assumption, where we have used~\eqref{Lie-q} in the penultimate equality.
\end{proof}

\begin{lemma}[Conformal Killing Vector]
\label{confLie}
With $\Sigma$, $n^a$, and $q_{ab}$ as above, if $\Lie_n q_{ab}=f\,q_{ab}$,
then there exists a unique conformal factor $\omega$, up to a constant
factor, such that $\Lie_{\overline{n}} \overline{q}_{ab}=0$, with
$\overline{n}^a$ and $\overline{q}_{ab}$ as in~\eqref{conformal}.
\end{lemma}

\begin{proof}
This is essentially the same as the result in Geroch~\cite{rG76} quoted above,
and the proof is similar.  We have
\begin{eqnarray}
\Lie_{\overline{n}}\,\overline{g_{ab}}
  &=& \overline{n}^c\,\Del_c \overline{g}_{ab}
	+ \overline{g}_{cb}\,\Del_a \overline{n}^c
	+ \overline{g}_{ac}\,\Del_b \overline{n}^c \nonumber \\
  &=& \omega^{-1}\,n^c\,\Del_c\, ( \omega^2\,g_{ab} )
	+ \omega^2\,g_{cb}\,\Del_a\,\left(\omega^{-1}\,n^c \right)
	+ \omega^2\,g_{ac}\,\Del_b\,\left( \omega^{-1}\,n^c \right)
	\nonumber \\
  &=& n^c\,\left( 2 (\Del_c\, \omega)\,g_{ab} - g_{cb}\,\Del_a\,\omega
	- g_{ac}\,\Del_b\,\omega \right) + \omega\,\Lie_n\,g_{ab}.
\end{eqnarray}
and pulling both sides back to $\Sigma$ results in
\begin{equation}
2 (n^c\, \Del_c\,\omega) \,q_{ab} =  -\omega\, f\,q_{ab}.
\end{equation}
Letting $\dot{\omega} = n^c\,\Del_c\,\omega $ the equation
simplifies to
\begin{equation}
\label{conformal_ode}
  \frac{\dot{\omega}}{\omega} = -\frac{f}{2}
\end{equation}
This ordinary differential equation will have a unique solution, up to a
constant factor, along each integral curve of $n^a$, yielding an $\omega$
such that the conformal transformation will result in $\Lie_{\overline{n}}
\overline{q_{ab}}=0$.
\end{proof}

These two lemmas immediately yield the following result:

\begin{theorem}[Covariant Derivative with conformal transformation]
\label{CovDerivII}
With $\Sigma$, $n^a$, and $q_{ab}$ as above, if $\Lie_n q_{ab}=f\,q_{ab}$ on
$\Sigma$, then the conformal pullback method produces a well-defined covariant
derivative, $D$, on $\Sigma$.
\end{theorem}

\begin{proof}
Since $\Lie_n q_{ab}=f\,q_{ab}$, Lemma~\ref{confLie} gives an $\omega$ such
that under the conformal transformation $\Lie_{\overline{n}}\overline{q_{ab}}=0$.
Now by Lemma~\ref{CovDerivI}, define the covariant derivative by
$D=\pback{\overline{\Del}}$.
\end{proof}

It is straightforward to verify that all of the examples considered so far
satisfy the conditions in Theorem~\ref{CovDerivII}.  The most interesting case
is the null cone~\eqref{cone}, for which
\begin{equation}
\Lie_{n}\,q_{ab} = \frac{2}{v}\,q_{ab}
\end{equation}
Substituting $f=2/v$ into~\eqref{conformal_ode} gives the ordinary
differential equation
\begin{equation}
\frac{\dot{\omega}}{\omega}
 = \frac{1}{\omega}\frac{\partial\omega}{\partial\,v}
 = -\frac{1}{v}.
\end{equation}
with solution
\begin{equation}
\omega = \frac{c}{v}
\end{equation}
where $c$ is a constant.  But
\begin{equation}
\frac{1}{r}\bigg|_{u=0} = \frac{\sqrt{2}}{v-u}\bigg|_{u=0} = \frac{\sqrt{2}}{v}
\end{equation}
and we see that rescaling the line element by $1/r^2$ leads to a well-defined
covariant derivative, as previously conjectured.

It is worth noting that if we regard the null cone as the \textit{unphysical}
space in Geroch's construction, then the corresponding ``physical''
stress-energy tensor only vanishes asymptotically to order 1, so that the
derivation of~\eqref{ODE} fails, although a similar result still holds, with a
different constant of proportionality.  We show in the next section that a
weaker condition on the stress-energy tensor is indeed sufficient for the
argument used here to work.

\section{Ricci Tensor}
\label{theorem-two}

The techniques adapted from Geroch's work on asymptotically flat spacetimes
have addressed the fundamental question: What are the conditions on a null
surface needed to construct a well-defined covariant derivative?  Either a
Killing normal vector, $\Lie_n q_{ab}=0$, or a conformal Killing vector,
$\Lie_n q_{ab}=f\,q_{ab}$, combined with a conformal transformation leads to a
well-defined covariant derivative on a null surface $\Sigma$ using the
pullback method.

One of the drawbacks of this construction is that the Lie derivative of the
metric must first be computed on $M$ and then pulled back to $\Sigma$ to test
the hypotheses of the theorems.  If the hypotheses are met, we return to $M$,
perform a conformal transformation if needed, compute $\Del$, then pull this
derivative back to $\Sigma$, giving $D$.  It would be nice if there was a test
to tell if the pullback led to a well-defined covariant derivative on $\Sigma$
and if a conformal transformation is required before pulling $\Del$ back to
$\Sigma$.  Again the work of Geroch leads to precisely such a condition

\begin{theorem}[Ricci Tensor and Covariant Derivative on $\Sigma$]
\label{connect}
Given a spacetime $\left(M,g_{ab}\right)$ containing a null surface
$\Sigma=\{u=0\}$, then if
$\pback{\Omega\left(R_{ab}-\widetilde{R}_{ab}\right)}=k\,q_{ab}$ , where
$\widetilde{g}_{ab}=\Omega^{-2}\,g_{ab}$, $\Omega = u$, and where $R_{ab}$ and
$\widetilde{R}_{ab}$ are the Ricci tensors of $g_{ab}$ and
$\widetilde{g}_{ab}$ respectively, then the conformal pullback method leads to
a well-defined connection on $\Sigma$.
\end{theorem}

\begin{proof}
From~\eqref{Ric} and~\eqref{LiegOmega},
\begin{equation}
\Omega\,\widetilde R_{ab}
 = \Omega\,R_{ab} + \frac{s-2}{2} \Lie_n g_{ab} + g_{ab} \Del^m \Del_m \Omega
	- \frac{s-1}{ \Omega} g_{ab} (\Del^{m} \Omega)(\Del_{m} \Omega).
\label{OmRic}
\end{equation}
Taking the trace of~\eqref{Ric} (using the metric $g$) yields
\begin{equation}
\frac{1}{\Omega} \> (\Del^{m} \Omega)(\Del_{m} \Omega)
  = \frac{2}{s} \> \Del^m \Del_m \Omega
	+ \frac{1}{s(s-1)} \>(\Omega R - \Omega^{-1}\widetilde{R})
\label{trRic}
\end{equation}
thus showing that the LHS admits a smooth limit to $\Sigma$ (since the last
term on the RHS does by assumption).  Setting $\Omega=u$ in~\eqref{OmRic} and
using~\eqref{trRic} and our hypotheses yields
\begin{equation}
\Lie_n q_{ab}
  = \pback{\Lie_n g_{ab}} \\
  = \left(
	-\frac{2k}{s-2} + \frac{2}{s} \Del^m \Del_m \Omega
	+ \frac{2}{s(s-2)} \Omega^{-1}\widetilde{R}
	\right) \Bigg|_{u=0}\,q_{ab}
\end{equation}
Theorem~\ref{CovDerivII} now implies that there is a covariant derivative $D$
on $\Sigma$.
\end{proof}

While the above derivation can be found in Geroch~\cite{rG76}, the
interpretation is quite different.  We start with a metric $g_{ab}$, a null
surface $\Sigma=\{u=0\}$ and a conformal factor defined by $\Omega=u$.  In the
sense of Geroch, we are creating an artificial ``physical'' space
$(\widetilde{M},\,\widetilde{g}_{ab})$ in order to determine if the pullback
method will result in a well-defined covariant derivative.  In the Geroch
approach, the boundary at null infinity was separated from $\widetilde{M}$ by
the conformal transformation in order to define some structure of the null
surface.  In our approach, one begins with the null surface and uses the
``physical'' space to define the covariant derivative.

\section{Further Examples}
\label{Examples}

The Schwarzschild metric in double-null Kruskal-Szekeres coordinates is given
by
\begin{equation}
ds^2 = -\frac{32\,m^3}{r}\,e^{-r/2m}\,du\,dv
	+ r^2\,d\theta^2 + r^2\,\sin^2\theta\,d\phi^2
\end{equation}
where $r$ is given implicitly by
\begin{equation}
u\,v = \left(1-\frac{r}{2m}\right)\,e^{r/2m}.
\end{equation}
The Schwarzschild metric is a vacuum solution of Einstein's equation, so the
Ricci tensor vanishes.  Considering the horizon at $u=0$ and using~\eqref{Ric}
with $\Omega=u$, the relevant components of the conformally related Ricci
tensor are given by
\begin{subequations}
\begin{eqnarray}
\widetilde{R}_{vv} &=& 0 \\
\widetilde{R}_{\theta \theta} &=& \dfrac{r}{m} \\
\widetilde{R}_{\phi \phi} &=& \dfrac{r}{m}\,\sin^2 \theta
\end{eqnarray}
\end{subequations}
and we have
\begin{equation}
\pback{\Omega\,(R_{ab}-\widetilde{R}_{ab})}
 = \pback{u\,(R_{ab}-\widetilde{R}_{ab})}
 = 0,
\end{equation}
trivially satisfying the necessary conditions of Theorem~\ref{connect}.
Straightforward computation verifies that the pullback connection constructed
from the conformally related metric
\begin{equation}
\widetilde{ds}^2
  = -\frac{32\,m^3}{r^3}\,e^{-r/2m}\,du\,dv + d\theta^2 + \sin^2\theta\,d\phi^2.
\end{equation}
is well-defined, and can therefore be used on the horizon.

More generally, a spherically symmetric space times has a line element of the
form
\begin{equation}
ds^2 = h\,du\,dv+r^2\,d\theta^2+r^2\,\sin^2\theta\,d\phi^2
\end{equation}
where $h$ and $r$ are both functions of the null coordinates $u$
and $v$.
As in previous examples, the null surface is chosen to be
$\Sigma=\{u=0\}$.  The nonzero components of the Ricci tensor are
\begin{subequations}
\begin{align}
R_{uu} &= -\dfrac{2\,\left(
	h\,\frac{\partial^2 r}{\partial u^2}
	- \frac{\partial h}{\partial u}\,\frac{\partial r}{\partial u}
	\right)}{h\,r} \\
R_{uv} &= \dfrac{r\,\frac{\partial h}{\partial u}\frac{\partial h}{\partial v}
	- h\,r\,\frac{\partial^2 h}{\partial v\, \partial u}
	- 2\,h^2\,\frac{\partial^2 r}{\partial v\, \partial u}}{h^2\,r} \\
R_{vv} &= -\dfrac{2\,\left(
	h\,\frac{\partial^2 r}{\partial v^2}
	- \frac{\partial h}{\partial v}\frac{\partial r}{\partial v}
	\right)} {h\,r} \\
R_{\theta \theta} &= -\dfrac{4\,\,r\,\frac{\partial^2 r}{\partial v\, \partial u}
	- h + 4\,\frac{\partial r}{\partial v} \frac{\partial r}{\partial u}}
	{h} \\
R_{\phi \phi} &= -\dfrac{\left(
	4\,r\,\frac{\partial^2 r}{\partial v\, \partial u}
	- h + 4\,\frac{\partial r}{\partial v}\frac{\partial r}{\partial  u}
	\right)\sin^2\theta}{h}
\end{align}
\end{subequations}
Setting $\Omega=u$, and again using~\eqref{Ric}, the conformally related Ricci
tensor has nonzero components
\begin{subequations}
\begin{align}
\widetilde{R}_{uu}
  &= \dfrac{2\,\left(
	-u\,h\,\frac{\partial^2 r}{\partial u^2}
	- r\,\frac{\partial h}{\partial u}
	+ u\,\frac{\partial h}{\partial u}\,\frac{\partial r}{\partial u}
	\right)}{u\,h\,r} \\
\widetilde{R}_{uv}
  &= \dfrac{u\,r\,\frac{\partial h}{\partial u}\frac{\partial h}{\partial v}
	- u\,h\,r\,\frac{\partial^2 h}{\partial v\, \partial u}
	+ 2\,h^2\,\frac{\partial r}{\partial v}
	- 2\,u\,\,h^2\,\frac{\partial^2 r}{\partial v\, \partial u}}
	{u\,h^2\,r} \\
\widetilde{R}_{vv}
  &= -\dfrac{2\,\left(h\,\frac{\partial^2 r}{\partial v^2}
	- \frac{\partial h}{\partial v}\frac{\partial r}{\partial v}\right)}
	{h\,r} \\
\widetilde{R}_{\theta \theta}
  &= -\dfrac{4\,u\,r\,\frac{\partial^2 r}{\partial v\, \partial u}- u\,h
	+ 4\,u\,\frac{\partial r}{\partial v}\frac{\partial r}{\partial u}
	- 8\,r\,\frac{\partial r}{\partial v}}{u\,h} \\
\widetilde{R}_{\phi \phi}
  &= -\dfrac{\left(4\,u\,r\,\frac{\partial^2 r}{\partial v\, \partial u} - u\,h 
	+ 4\,u\,\frac{\partial r}{\partial v}\frac{\partial r}{\partial u}
	- 8\,r\,\frac{\partial r}{\partial v}\right)\sin^2\theta}{u\,h}
\end{align}
\end{subequations}
Computing the pullback of
$\Omega(R_{ab}-\widetilde{R}_{ab})=u\,(R_{ab}-\widetilde{R}_{ab})$,
the relevant components are
\begin{subequations}
\begin{align}
u(R_{vv}-\widetilde{R}_{vv}) &\longrightarrow 0  \\
u(R_{\theta \theta}-\widetilde{R}_{\theta \theta})
  &\longrightarrow \dfrac{-8\,r\,\frac{\partial r}{\partial v}}{h}
   = \left(\dfrac{-8\,\frac{\partial r}{\partial v}}{r\,h}\right)\,r^2 \\
u(R_{\phi \phi}-\widetilde{R}_{\phi \phi})
  &\longrightarrow \dfrac{-8\,r\,\frac{\partial r}{\partial v}}{h}
	\,\sin^2 \theta
   = \left(\dfrac{-8\,\frac{\partial r}{\partial v}}{r\,h}\right)
	\,r^2\,\sin^2 \theta
\end{align}
\end{subequations}
so that
\begin{equation}
\pback{\Omega(R_{ab}-\widetilde{R}_{ab})}
  = \left(\dfrac{-8\,\frac{\partial r}{\partial v}}{r(0,v)\,h(0,v)}\right)
	\,q_{ab}
\end{equation}
Thus the conditions of Theorem~\ref{connect} are satisfied.  More importantly,
the conformal transformation $\overline{g}_{ab}=(\omega^2)\,g_{ab}$ with
$\omega=r(u,v)$ will yield $\Lie_{\overline{n}}\overline{q}_{ab} = 0$, the
condition needed to use the pullback to produce a well-defined covariant
derivative on $\Sigma=\{u=0\}$.

Again, straightforward computation verifies that the pullback connection
constructed from the conformally related metric is indeed well-defined on
$\Sigma$.

\section{Summary}
\label{Summary}

Due to the degenerate metric, working with null surfaces offers some very
challenging obstacles, since traditional tools such as Christoffel symbols are
not defined.  Gauss decomposition fails, since there are non-zero null vectors
both tangent and perpendicular to the hypersurface.  The work of Duggal and
Benjacu attempts to overcome this difficulty by defining a screen manifold and
a lightlike transversal vector bundle to decompose the manifold and the null
hypersurface.  However, even with all of this structure, there are elementary
examples, such as the null cone, that do not satisfy the hypotheses necessary
to produce a covariant derivative independent of the screen.

An alternate construction uses the pullback to define the covariant
derivative.  Care must be taken when using this technique, since the resulting
derivative operator may not be well defined; this technique will only work as
long as the null vector field is a Killing vector field.

Motivated by the work of Geroch on asymptotically flat spacetimes, conformal
transformations were used not only to give a well-defined derivative on null
hypersurfaces, but also to provide a test to determine whether the null
surface admits such a definition. 

Finally, the conformal pullback method was shown to work at the horizon of
the Schwarzschild geometry, and more generally for any spherically symmetric
spacetime.

Further work is needed to understand the implications of this construction.
For example, in the case of the null cone, the conformal transformation
results in a null cylinder.  What does it mean to use the cone's covariant
derivative operator on a sphere?  And are there non-symmetric null surfaces on
which this construction works?

Finally, we remark that there does not appear to be a similar technique for
Riemannian spaces, but since traditional Gauss decomposition works, this
technique is not needed in that case.

\section*{Acknowledgments}
This paper is based on work submitted by DH in partial fulfillment of the
degree requirements for his Ph.D.\ in Mathematics at Oregon State
University~\cite{HickThesis}.

\end{document}